\documentclass{book} 

\pdfoutput=1

\usepackage{array}
\newcounter{subsectionsinica}

\newcommand{\sectionsinica}[2]{
\setcounter{subsectionsinica}{0}
\refstepcounter{chapter} %
\label{#2} %
\par\noindent {\normalfont \Large\bfseries  \arabic{chapter}.  #1\\}
}

\newcommand{\subsectionsinica}[1]{
\stepcounter{subsectionsinica} %
\par\noindent {\bf \arabic{chapter}.\arabic{subsectionsinica} #1\\}
}

\def\isblinded{false} %
\usepackage{amssymb}
\usepackage{amsmath}
\allowdisplaybreaks[1] %
\usepackage{bbold} %

\usepackage{wasysym} %
\usepackage{afterpage}%
\usepackage{url} %

\usepackage{apacite}
\newcommand{\citealp}[1]{\citeNP{#1}}
\newcommand{\citet}[1]{\citeA{#1}}
\newcommand{\citep}[1]{\cite{#1}}
\renewcommand{\BBAA}{and}

\ifx\pdfoutput\undefined
  \usepackage[dvips]{graphicx}
\else
  \usepackage[pdftex]{graphicx}
\fi
\graphicspath{figures/}

\def\ba#1\ea{\begin{align*}#1\end{align*}} %
\def\banum#1\eanum{\begin{align}#1\end{align}} %

\newtheorem{lemma}{Lemma}

\newtheorem{theorem}[lemma]{Theorem}
\newtheorem{definition}[lemma]{Definition}

\newcounter{construction}
\newtheorem{construction}{Construction}

\newcommand{\ulesqed}{\hfill\smiley}
\newenvironment{proof}{\par\noindent{\em Proof. }}{\hfill\ulesqed\\[2mm]}

\newcommand{\RR}{\mathbb{R}} %

\newcommand{\R}{\RR}

\renewcommand{\epsilon}{\ensuremath{\varepsilon}}

\renewcommand{\Pr}{\pi}
\newcommand{\rhoperp}{{\rho_\perp}}
\newcommand{\rperp}{{r_\perp}}

\newcommand{\rgeo}{\ensuremath{R_{\text{geo}}}}
\newcommand{\rfieller}{\ensuremath{R_{\text{Fieller}}}}
\newcommand{\rcons}{\ensuremath{R_{\text{cons}}}}
\newcommand{\rgen}{\ensuremath{R_{\text{gen}}}}

\begin{document}
\markright{
} 
\markboth{\hfill{\footnotesize\rm ULRIKE VON LUXBURG AND VOLKER H. FRANZ
}\hfill}
{\hfill {\footnotesize\rm CONFIDENCE SETS FOR RATIOS} \hfill}
\renewcommand{\thefootnote}{}
$\ $\par
\fontsize{10.95}{14pt plus.8pt minus .6pt}\selectfont
\vspace{0.8pc}
\centerline{\large\bf A Geometric Approach to Confidence Sets for Ratios: }
\vspace{2pt}
\centerline{\large\bf Fieller's Theorem, Generalizations, and Bootstrap}
\vspace{.4cm}
\centerline{Ulrike von Luxburg and Volker H. Franz}
\vspace{.4cm}
\centerline{\it Max Planck Institute for Biological Cybernetics,
  T{\"u}bingen, Germany}
\vspace{2pt}
\centerline{\it 
Justus--Liebig--Universit{\"a}t Giessen, Germany}
\vspace{.55cm}
\fontsize{9}{11.5pt plus.8pt minus .6pt}\selectfont

\begin{quotation}
\noindent {\it Abstract:} 
We present a geometric method to determine confidence sets for
the ratio $E(Y)/E(X)$ of the means of random variables $X$ and
$Y$. This method reduces the problem of constructing confidence sets
for the ratio of two random variables to the problem of constructing
confidence sets for the means of one-dimensional random variables. It
is valid in a large variety of circumstances.  In the case of normally
distributed random variables, the so constructed confidence sets
co\-incide with the standard Fieller confidence sets. Generalizations
of our construction lead to definitions of exact and conservative confidence sets for
very general classes of distributions, provided the joint expectation of
$(X,Y)$ exists and the linear combinations of the form $aX + bY$ are
well-behaved. Finally, our geometric method allows to derive a very
simple bootstrap approach for constructing conservative confidence sets
for ratios which perform favorably in certain situations, in
particular in the asymmetric heavy-tailed regime. 
\par

\end{quotation}\par

\fontsize{10.95}{14pt plus.8pt minus .6pt}\selectfont
\sectionsinica{Introduction}{sec-intro} %
In many practical applications we encounter the problem of estimating
the ratio of two random variables $X$ and $Y$.  This could, for
example, be the case if we want to know how large one quantity is
relative to the other, or if we want to estimate at which position a
regression line intersects the abscissa (e.g.,
\citet{Miller_86,Buonaccorsi_01}; see also
\ifx\isblinded\istrue
BLINDED REFERENCE
\else
\citet{Franz_fiellersub} 
\fi
for many references to practical
studies involving ratios).
While it is straightforward to construct an estimator for $E(Y)/E(X)$
by dividing the two sample means of $X$ and $Y$, it is not obvious how
confidence regions for this estimator can be defined. In the case
where $X$ and $Y$ are jointly normally distributed, an exact solution
to this problem has been derived by
\citet{Fieller32,Fieller40,Fieller44,Fieller54}; for more detailed
discussions see \citet{Kendall61}, \citet{Finney_78},
\citet{Miller_86}, and \citet{Buonaccorsi_01}.  But in applications,
practitioners often do not use Fieller's results and apply ad-hoc
methods instead.  Perhaps the main reason is that Fieller's  confidence
regions do not look like "normal" confidence intervals and are often
perceived as counter-intuitive. In benign cases they form an interval
which is not symmetric around the estimator, while in worse cases the
confidence region consists of two disjoint unbounded intervals, or
even of the whole real line.
Especially the latter case is highly unusual as the confidence region does not exclude any value at all --- certainly
not what one would expect from a well--behaved confidence
region. However, different researchers
\citep{Gleser_Hwang_87,Koschat_87,Hwang_95} have shown that any method
which is not able to generate such unbounded confidence limits for a
ratio leads to arbitrary large deviations from the intended
confidence level. For a discussion of the conditional confidence level, given
that the Fieller confidence limits are bounded,  see
\citet{Buonaccorsi_Iyer_84}.\\

There have been several approaches to present Fieller's methods in a
more intuitive way.  Especially remarkable are the ones which rely on
geometric arguments. \citet{Milliken82} attempted a geometric proof
for Fieller's result in the case where $X$ and $Y$ are independent
normally distributed random variables. Unfortunately, his proof
contained an error which led him to the wrong conclusion that
Fieller's confidence regions were too conservative.  Later, his proof
was corrected and simplified by \citet{Guiard89}. He considers the
case that $X$ and $Y$ are jointly normally distributed according to
$(X,Y) \sim N(\mu, \sigma^2 V)$, where the mean $\mu$ and the scale
$\sigma^2$ of the covariance are unknown, but the covariance matrix
$V$ is known. Guiard presents a geometric construction of confidence
regions, and then shows by an elegant comparison to a likelihood ratio
test that the constructed regions are exact and coincide with
Fieller's solution. The drawback of his proof is that it only works in
the case where the covariance matrix $V$ is known, which  in practice is usually
not the case. Moreover, although the confidence sets are
constructed by a geometric procedure, Guiard's proof relies on
properties of the likelihood ratio test and does not give geometric
insights into why the construction is correct. \\

In this article we derive several simple geometric constructions for
exact confidence sets for ratios. Our construction coincides with
Guiard's if $(X,Y)$ are normally distributed with known
covariance matrix $V$, but it is also valid in the case where $V$ is
unknown. Our proof techniques are remarkably simple and purely
geometric. The understanding gained by our  approach  then allows to
extend the geometric construction from normally distributed random variables to
more general classes of distributions.  While it is relatively
straightforward to define confidence sets for elliptically symmetric
distributions, another extension  leads to a completely new
construction of confidence sets for ratios which is exact for a very
large class of distributions. Essentially, the only assumptions we
have to make is that the means of $X$ and $Y$ exist and that it is
possible to construct exact confidence sets for the mean of linear
combinations of the form $a_1 X + a_2 Y$. To our knowledge, this is
the first definition of {\em exact} confidence sets for ratios of very
general classes of distributions. Finally, using the geometric
insights also leads to a simple bootstrap procedure for confidence
sets for ratios. This method is particularly well-suited for highly
asymmetric and heavy-tailed distributions. \\

\subsectionsinica{Definitions and notation} %
We will always consider the following situation. We are given a sample
of $n$ pairs $Z_i:= (X_i,Y_i)_{i=1,...,n}$ drawn independently
according to some underlying distribution. In the first part we will
always assume that this joint distribution is a 2-dimensional normal distribution $N(\mu, C)$
with mean $\mu=(\mu_1,\mu_2)$ and covariance matrix $C =
\big( \begin{smallmatrix} c_{11} & c_{12} \\ c_{21} & c_{22}
\end{smallmatrix} \big) $.  We assume
that both $\mu$ and $C$ are unknown. Later we will also study more
general classes of distributions. Our goal will be to estimate the
ratio $\rho := \mu_2/\mu_1$ and construct confidence sets for this ratio. 
To estimate the unknown mean and the covariance matrix
we will use the standard estimators: the means are estimated by
\banum \label{eq-est-mean}
\hat\mu_1 := \frac{1}{n}\sum_{i=1}^n X_i  \;\;\text{and}\;\;
\hat\mu_2 := \frac{1}{n}\sum_{i=1}^n Y_i,
\eanum
and the estimated  covariance matrix $\hat{C} = 
\big( \begin{smallmatrix} \hat{c}_{11} & \hat{c}_{12} \\ \hat{c}_{21} & \hat{c}_{22}
\end{smallmatrix} \big)$ has the entries
\banum \label{eq-est-cov}
&  \hat{c}_{11} := \frac{1}{{n}}\frac{1}{n-1} \sum_{i=1}^n (X_i - \hat\mu_1)^2
\;\; \text{and}\;\;
\hat{c}_{22} := \frac{1}{{n}}\frac{1}{n-1} \sum_{i=1}^n (Y_i - \hat\mu_2)^2 \\
\nonumber &  \hat{c}_{12} := \hat{c}_{21} =\frac{1}{{n}}\frac{1}{n-1} \sum_{i=1}^n (X_i - \hat\mu_1)(Y_i -
\hat\mu_2).
\eanum
Note that we rescaled the estimators $\hat{c_{ij}}$ by $1/{n}$ to
reflect the variability of the estimators $\hat\mu_1$ and
$\hat\mu_2$. This will be convenient later on. As estimator for the
ratio $\rho = \mu_2/\mu_1$ we use 
$\hat\rho := {\hat\mu_2}/{\hat\mu_1}$.
Note that our goal is to estimate $E(Y)/E(X)$ and not $E(Y/X)$. In fact, if $X$
and $Y$ are normally distributed, the
latter quantity does not even exist. 
As in this situation  the estimators $\hat\mu_1$ and
$\hat\mu_2$ are normally distributed as well, we can also see that  the estimator
$\hat\rho$ 
cannot be
unbiased, as its expectation $E(\hat\rho) = E(\hat\mu_2/\hat\mu_1)$
simply does not exist.
For more discussion on the bias of the estimator $\hat\rho$ see
\citet{Beale_62,Tin_65,Durbin_59,Rao_81,Miller_86} and
\citet{Dalabehera_Sahoo_95}. \\

For $\alpha \in ]0,1[$, a {\em confidence
set} (or {\em confidence region}) of level $1-\alpha$ for a
parameter $\theta \in \Theta$ 
is defined to be a set $R$ constructed from the sample such
that for all $\theta \in \Theta$ it holds that $P_\theta(\theta \in R)
\geq 1-\alpha$. 
If this statement holds with equality, then the confidence
set $R$ is called {\em exact}, otherwise it is called {\em
  conservative}.  If the statement $P_\theta(\theta \in R) = 1-\alpha$ only
holds in the limit for the sample size $n \to \infty$, the confidence
set $R$ is called {\em asymptotically exact}. A confidence interval
$[l,u]$ is
called {\em equal-tailed} if $P_\theta(\theta < l) = P_\theta(\theta > u)$. It is
called {\em symmetric around} $\hat\theta$ if it has the form
$[\hat\theta -q, \hat\theta +q]$. 
For general background
reading about confidence sets we refer to Chapter 20 of
\citet{Kendall61}, Section 5.2 of \citet{Schervish95}, and Chapter 4
of \citet{ShaTu95} . 
For a real-valued random variable with distribution function $F$ and a
number $\alpha \in ]0,1[$, the $\alpha$-quantile of $F$ is defined as the
smallest number $x$ such that $F(x) = \alpha$. We will denote this quantile
by $q(F,\alpha)$. In the special case where $F$ is induced by the Student-t
distribution with $f$ degrees
of freedom, we will denote the quantile by $q(t_f,\alpha)$.  \\

Many of the geometric arguments in this paper will be based on
orthogonal projections of the
two-dimensional plane to a one-dimensional subspace.
In the two-dimensional plane, we define the line $L_\rho$ through the origin with
slope $\rho$ and the line $L_\rhoperp$ orthogonal to $L_\rho$ by 
\ba
&L_{\rho} := \{(x,y) \in \RR^2 | \; y = \rho x\} \\
& L_\rhoperp := \{ (x,y) \in \RR^2 | \; y = (-1/\rho) x \}. 
\ea
For an arbitrary unit vector 
$a = (a_1,a_2)'  \in \RR^2$ let 
\ba
\Pr_a:\RR^2 \to \RR, \; x \mapsto a'x = a_1x_1 + a_2x_2
\ea
be the orthogonal projection of the two-dimensional plane on the
one-di\-men\-sio\-nal subspace spanned by $a$, that is on the line
$L_r$ with slope $r = a_2/a_1$. We will also 
write $\pi_r$ for the projection on $L_r$, and
$\pi_\rperp$ for the projection on the line $L_\rperp$. \\

 Let $C
\in \RR^{2 \times 2}$ be a covariance matrix (i.e., positive definite
and symmetric) with eigenvectors $v_1,v_2 \in \RR^2$ and
eigenvalues $\lambda_1,\lambda_2 \in \RR$. Consider the ellipse
centered at some point $\mu \in \RR^2$ such that its principal
axes have the directions of $v_1,v_2$ and have lengths $q
\sqrt{\lambda_1}$ and $q\sqrt{\lambda_2}$ for some $q>0$. We
denote this ellipse by $E(C,\mu,q)$ and call it the {covariance
ellipse} corresponding to $C$ centered at $\mu$ and scaled with
parameter $q$. This ellipse can also be described as the set of
points $z \in \RR^2$ which
satisfy the ellipse equation $(z-\mu)'C^{-1}(z-\mu) = q^2$. \\

\sectionsinica{Exact confidence regions for
  normally distributed random variables}{sec-exact}
Let us start with a few geometric
observations.  For given $\mu = (\mu_1,\mu_2) \in \RR^2$, the ratio
$\rho = \mu_2/\mu_1$ can be depicted as the slope of the line $L_\rho$
in the two-dimensional plane which passes both through the origin and
the point $(\mu_1,\mu_2)$.  Similarly, the estimated ratio $\hat\rho$
is given as the slope of the line through the origin and the point
$\hat\mu = (\hat\mu_1, \hat\mu_2)$ (cf. Figure
\ref{fig-fieller-intro}).
\begin{figure}[!bt]
  \begin{center}
 \includegraphics[width=0.8\textwidth,height=0.3\textheight]{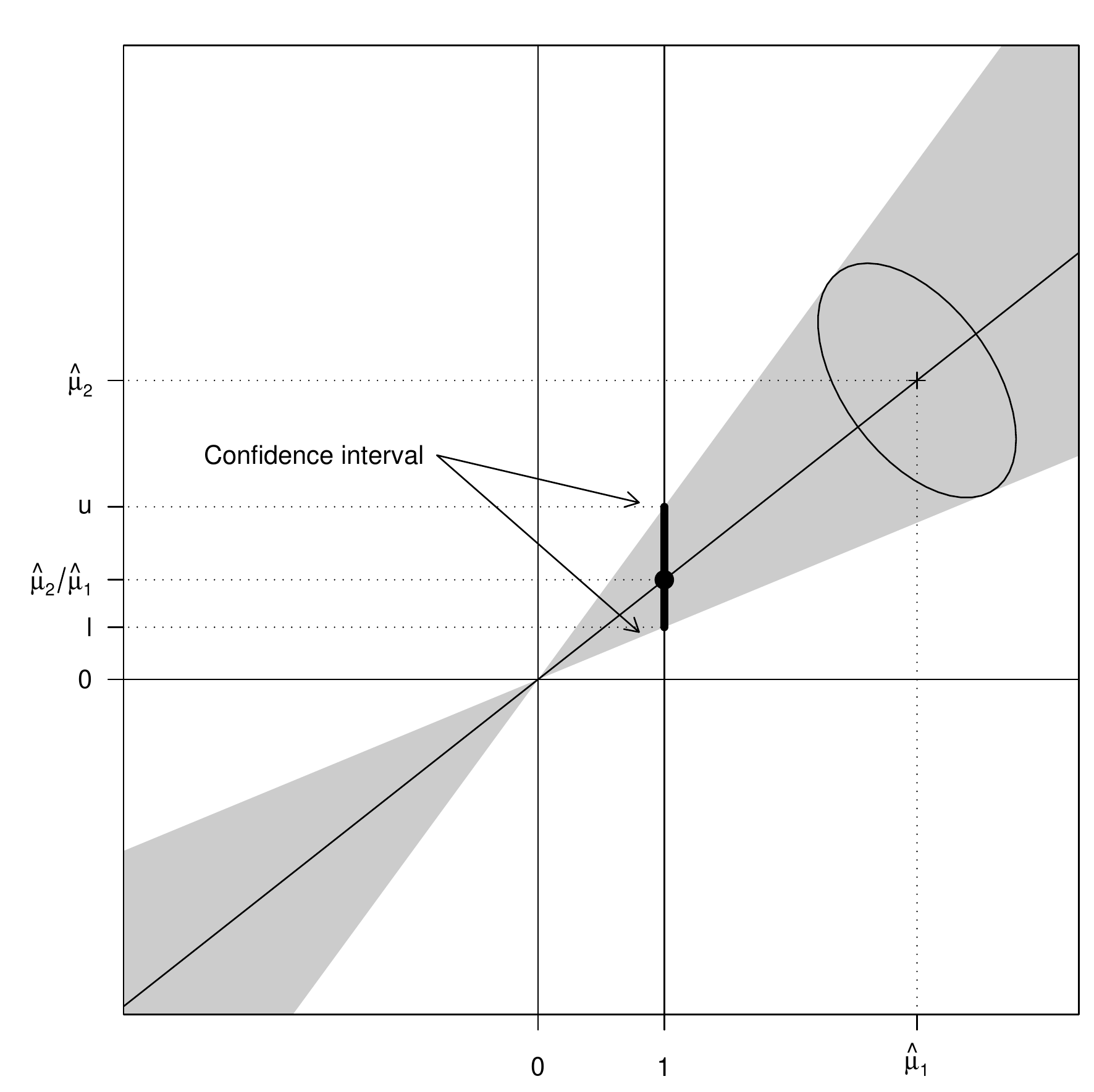}
\caption{Geometric principles. The ratio $\hat\mu_2/\hat\mu_1$ can be depicted as the slope of
the line through the points $(0,0)$ and $(\hat\mu_1,\hat\mu_2)$. The
ratios inside an interval $[l,u]$ correspond to the slopes of all
lines in the wedge spanned by the lines with slopes $l$ and
$u$. For a given wedge, the corresponding interval $[l,u]$ can be
obtained by intersecting the wedge with the line $x=1$. %
}
  \label{fig-fieller-intro}
  \end{center}
\end{figure}
Assume that we are given a confidence interval $R=[l,u] \subset \RR$ that
contains the estimator $\hat\rho$.   The lower and upper limits of this interval
correspond to the slopes of the two lines passing
through the origin and the points $(1,l)$ and $(1,u)$,
respectively. Let $W$ denote the wedge enclosed by those two
lines. The slopes of the lines inside the wedge 
exactly correspond to the ratios inside the interval $R$. The other
way round,  the interval $[l,u]$ can be reconstructed from the
wedge as the intersection of the wedge with the line $x=1$
(cf. Figure \ref{fig-fieller-intro}). \\

\subsectionsinica{Geometric construction of exact confidence sets} 
In the following we want to construct an appropriate wedge containing
$\hat\mu$ such that the region obtained by intersection with the line
$x=1$ yields an exact confidence region for $\rho$ of level
$1-\alpha$. This wedge will be constructed as the smallest wedge
containing a certain ellipse around the estimated mean
$(\hat\mu_1,\hat\mu_2)$. We will see that depending on the position of
the ellipse, we have to distinguish between three different cases
called ``{\em bounded}'', ``{\em exclusive unbounded}'', and {\em
  ``completely unbounded}''.  For an illustration see Figure
\ref{fig-fieller-fullexample}.

\begin{construction} \label{construction-fieller-geometric}
{\bf (Geometric construction of exact confidence regions $\rgeo$  for
  $\boldsymbol{\rho}$ in case of normal distributions)}
\begin{enumerate}
\item Estimate the means $\hat\mu_1$ and $\hat\mu_2$ according to
  Equation \eqref{eq-est-mean}, the covariance matrix $\hat{C}$ according to
  Equation \eqref{eq-est-cov}.
\item Define the real number $q$ as $q(t_{n-1},1-{\alpha}/{2})$. That
  is, $q$ is the $(1-{\alpha}/{2})$-quantile of the Student-$t$
  distribution with $n-1$ degrees of freedom.
\item In the two-dimensional plane, plot the ellipse
  $E = E(\hat{C},\hat\mu,q)$ centered at the estimated joint
  mean  $\hat\mu=(\hat\mu_1,\hat\mu_2)$, with shape according to the estimated
  covariance matrix $\hat{C}$, and scaled by the number $q$
  computed in the step before.
\item Depending on the position of the ellipse, distinguish between the
  following cases (see Figure \ref{fig-fieller-fullexample}).
  \begin{enumerate}
  \item If $(0,0) \text{ not inside } E$, construct the two tangents
    to $E$ through the origin $(0,0)$ and let $W$ be the wedge
    enclosed by those tangents.  Define the region $\rgeo$ as the
    intersection of $W$ with the line $x=1$. Depending on whether the
    $y$-axis lies inside $W$ or not, this results in an exclusive
    unbounded or a bounded confidence region.
     \item
       If  $(0,0) \text{ inside } E$, choose the confidence region as  $\rgeo=
       ]-\infty,\infty[$ (completely unbounded case). 
     \end{enumerate}
\end{enumerate}
\end{construction}

Let us give some 
intuitive reasons why the three cases and the form of the confidence
sets make sense. In the first case, the 
denominator $\hat \mu_1$ is significantly different from $0$. Here 
we do not expect any difficulties from dividing by $\hat \mu_1$ as the
denominator is ``safely away from 0''. Our uncertainty about the value
of $\rho$ is restricted to some interval around $\rho$, which
corresponds to the bounded case. To relate this to the geometric
construction, observe that the denominator is significantly different
from 0 if and only if the ellipse $E$ does not touch the $y$-axis.  
The situation is more complicated if the denominator is not
significantly different from $0$, that is the ellipse intersects with
the $y$-axis. As we divide by a number potentially close to $0$, we
cannot control the absolute value of the outcome, which
might become arbitrarily large, nor can we be sure about its sign.
Hence, regions of the form $]-\infty,c_1]$ and $[c_2, \infty[$ should
be part of the confidence region.  If, additionally, we are confident
that the numerator is not too small, then we expect that $\rho$ is
not very close to 0.  This is reflected by the ``exclusive
unbounded case''.  If, on the other hand, the numerator is not
significantly different from 0, then we cannot guarantee for anything:
when dividing $0/0$ any outcome is conceivable. Here the confidence set
should coincide with the whole real line, which is the ``completely
unbounded'' case.
\\

\begin{figure}[!bt]
  \begin{center}
 \includegraphics[height=0.3\textheight,width=0.7\textwidth]{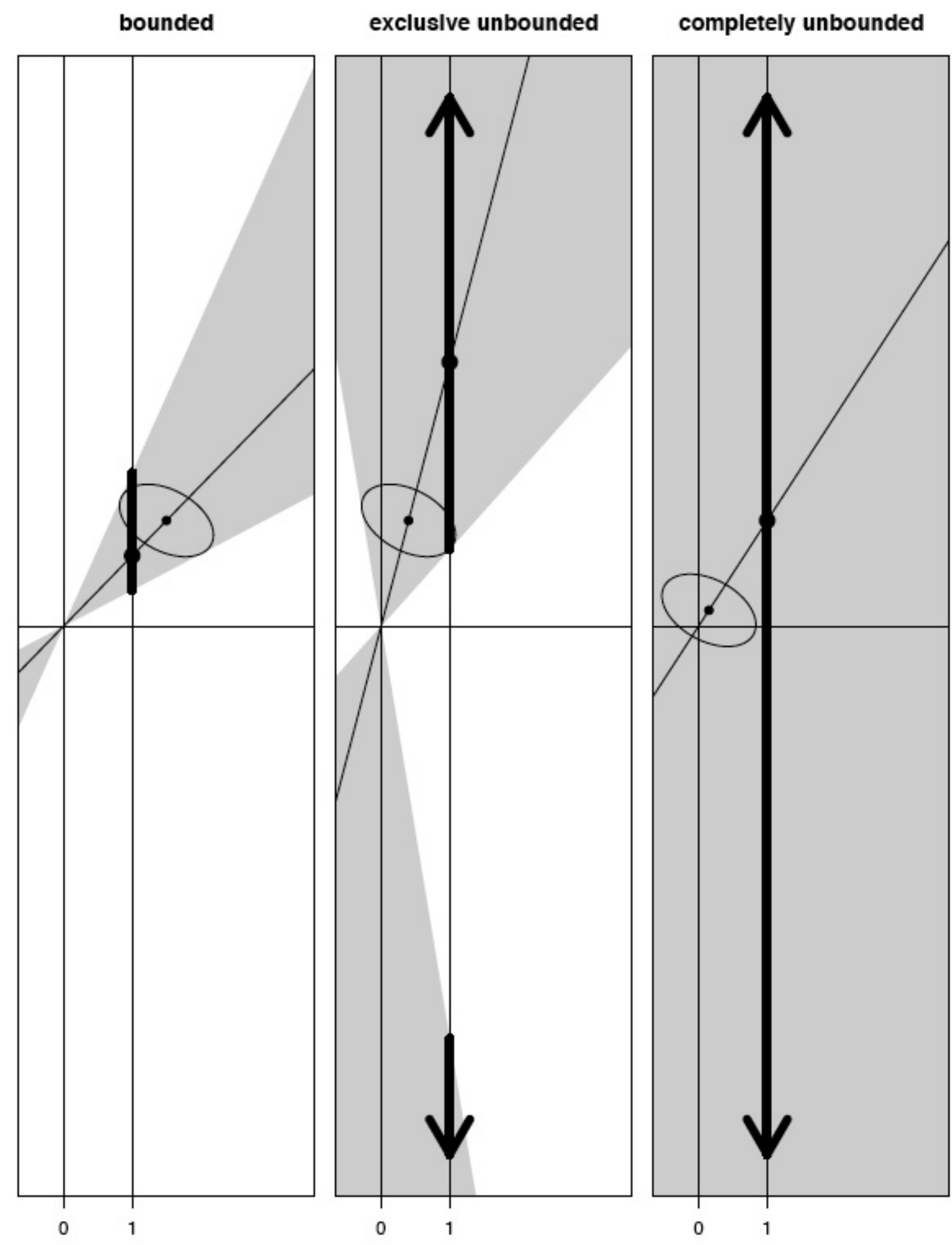}
\caption{The three cases in the construction of the confidence set
  $\rgeo$: the bounded case where the ellipse does
  not intersect the $y$-axis, the exclusive
  unbounded case, where the ellipse intersects the $y$-axis but does
  not contain the origin, and the completely unbounded case,
  where the ellipse contains the origin. }
  \label{fig-fieller-fullexample}
  \end{center}
\end{figure}

\begin{theorem}[$\rgeo$ is an exact confidence set for $\rho$]\sloppy
  \label{th-we-are-exact}
Let $(X_i,Y_i)_{i=1,...,n}$ be an i.i.d. sample
drawn from the distribution $N(\mu,C)$ with unknown $\mu$ and $C$, and let $\rgeo$ be the
regions constructed according to Construction \ref{construction-fieller-geometric}. Then $\rgeo$
is an exact confidence region of level $1-\alpha$ for $\rho$, that is
 for all $\mu$ and $C$ we have $P(\rho \in \rgeo) = 1-\alpha$.
\end{theorem}
\begin{figure}[!bt]
  \begin{center}
 \includegraphics[width=0.7\textwidth]{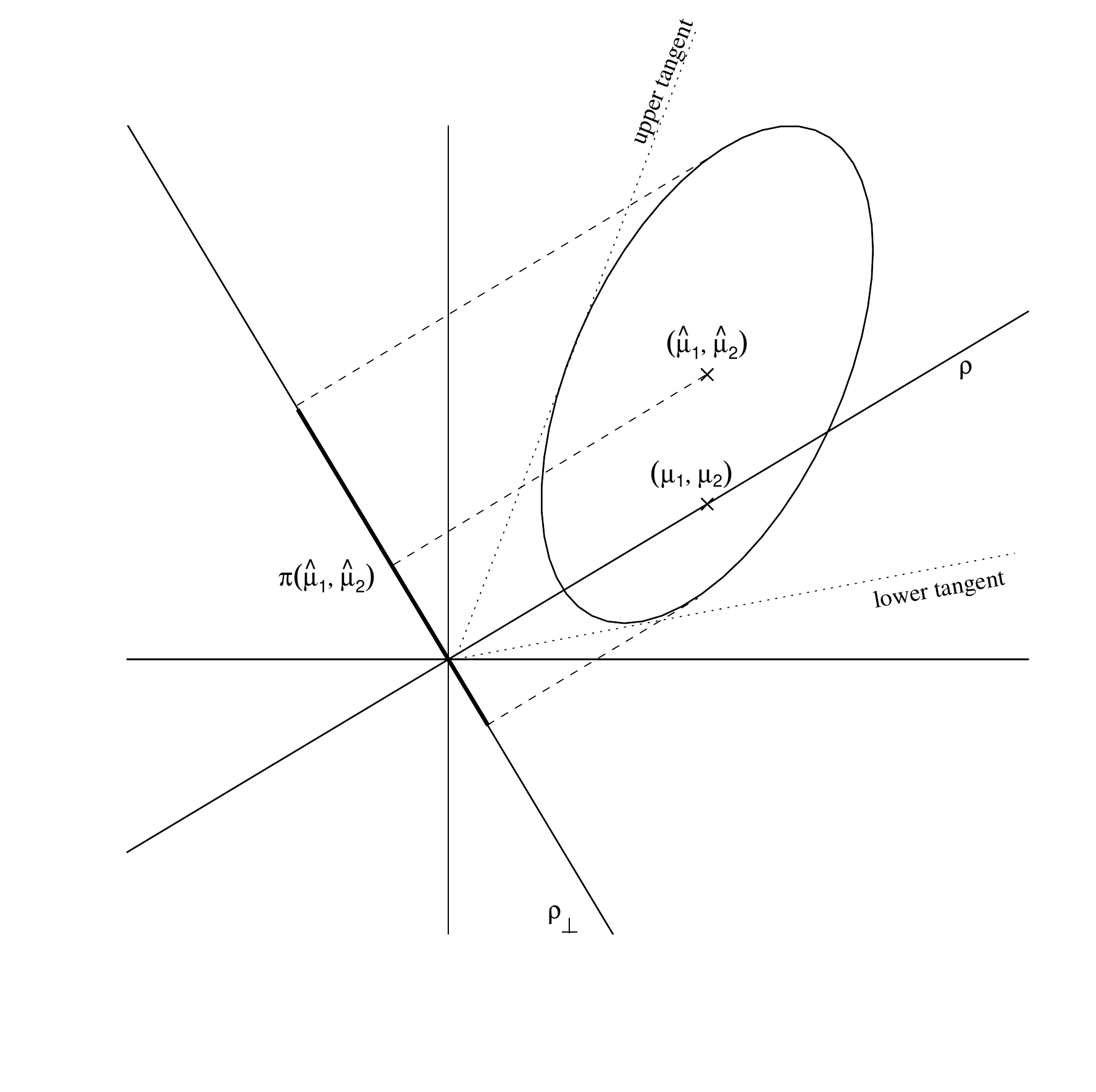}
\caption{Projection of the ellipse $E$ on the subspace spanned by
$\rho_{\perp}$ (see proof of Theorem \ref{th-we-are-exact}). }
  \label{fig-fieller-proof}
  \end{center}
\end{figure}
\begin{proof}
Let $a = (a_1,a_2)'  \in \RR^2$ be an arbitrary
unit vector. 
We denote by $U := \Pr_a(X,Y)$ the projection of the joint random
variable $(X,Y)$ on the subspace spanned by $a$. 
Then $U$ is distributed according
to $N(a'\mu, a'Ca)$. The independent sample points $(X_i,Y_i)_{i=1,...n}$ are
mapped by $\Pr_a$ to independent sample points $(U_i)_{i=1,...,n}$. 
It is easy to see that the length of the interval $I := \Pr_a(E)$ is
$2 q (a'
\hat{C} a)^{1/2}$. 
Taking into account the choice of the scaling factor $q$ in
Construction \ref{construction-fieller-geometric} as the
$(1-{\alpha}/{2})$-quantile of the Student-$t$ distribution, by
the normality assumption on $(X,Y)$ we can now conclude that the
projected ellipse $\Pr_a(E)$ is a $(1-\alpha)$-confidence interval for
the mean $\pi_a(\mu)$ of the projected random variables:
\ba
1-\alpha
&= P\left(\Pr_a(\mu) \in [\Pr_a(\hat\mu) - q (a'\hat{C}a)^{1/2},\;\Pr_a(\hat\mu) + q (a'\hat{C}a)^{1/2}  ]\right)\\
&= P\left(\Pr_a(\mu) \in \Pr_a(E)\right).
\ea
This equation is true for all unit vectors $a$. 
Now we want to consider the particular projection $\pi_\rhoperp$ on
the line $L_\rhoperp$ (that is, we choose \sloppy $a = (\rho/\sqrt{1+\rho^2},
{-1}/\sqrt{1+\rho^2})$). 
Showing that $ \Pr_\rhoperp(\mu) \in
\Pr_\rhoperp(E) \iff \rho \in \rgeo$ will complete our proof. 
As in the construction of $\rgeo$ we distinguish two cases. If the
origin is not inside the ellipse $E$ we can
construct the wedge $W$ as described in the construction of $\rgeo$. In
this case we have the following geometric equivalences (see Figure
\ref{fig-fieller-proof}):
\ba \Pr_\rhoperp(\mu) \in \Pr_\rhoperp(E) 
&  \iff 0 \in \Pr_\rhoperp(E) 
& \iff E \cap L_\rho \neq \emptyset 
& \iff L_\rho \subset W 
& \iff \rho \in \rgeo.
\ea
In the second case, the origin  is inside in the ellipse $E$.
In this case it is clear that $\Pr_\rhoperp(\mu) = 0$ is
always inside $\Pr_\rhoperp(E)$. On the other hand, by
definition the region $\rgeo$ coincides with $]{-\infty},\infty[$ in
this case, and thus $\rho \in \rgeo$ is true. 
\mbox{}
\end{proof}

\subsectionsinica{Comparison to Fieller's confidence sets} \label{sec-proof}
Theorem \ref{th-we-are-exact} shows that the confidence regions $\rgeo$
obtained by Construction~\ref{construction-fieller-geometric} are
exact confidence regions. Now we want to compare them to the
classic confidence sets constructed by 
\citet{Fieller32,Fieller40,Fieller44,Fieller54}.  To this end
let us first state Fieller's result according to Subsection 4, p.
176-177 of \citep{Fieller54}. 
We reformulate his definition in
our notation:

\begin{definition} \label{construction-fieller} 
{\bf(Fieller's confidence regions for $\boldsymbol\rho$ in case of normal distributions) }
Compute the quantities
\ba
& q_{\text{exclusive}}^2:= \frac{\hat\mu_1^2}{\hat{c}_{11}}
\;\; \text{  and } \;\;
q_{\text{complete}}^2:= \frac{
\hat\mu_2^2 \hat{c}_{11}
- 2 \hat\mu_1 \hat\mu_2 \hat{c}_{12}
+\hat\mu_1^2 \hat{c}_{22}}{\hat{c}_{11}\hat{c}_{22} -
\hat{c}_{12}^2} \;\; \text{  and } \;\;\\
& l_{1,2} = 
\frac{1}{\hat\mu_1^2 - q^2 \hat{c}_{11}} 
\bigg(
(\hat\mu_1 \hat\mu_2 - q^2 \hat{c}_{12}) 
 \pm
\sqrt{
(\hat\mu_1 \hat\mu_2 - q^2 \hat{c}_{12})^2 - 
(\hat\mu_1^2 - q^2 \hat{c}_{11})
(\hat\mu_2^2 - q^2 \hat{c}_{22}) }
\bigg)
\ea 
with  $q$ as in the
definition of the confidence regions $\rgeo$.
Then define the confidence set $\rfieller$ of level $1-\alpha$ for the ratio $\rho$ as follows:
\ba
\rfieller =
\begin{cases}
]-\infty,\infty[
& \text{ if }
q_{\text{complete}}^2 \leq q^2 \\
]-\infty, \min\{l_1,l_2\}] \;\;\; \cup \;\;  [\max\{l_1,l_2\},\infty[
& \text{ if }
q_{\text{exclusive}}^2 < q^2 < q_{\text{complete}}^2 \\
[ \min\{l_1,l_2\}, \max\{l_1,l_2\}]
& \text{ otherwise } \\
\end{cases}
\ea
Those three cases result in completely unbounded, exclusive
unbounded, and  bounded confidence sets, respectively. 
\end{definition}
\begin{theorem}[Fieller]
\label{th-fieller}
Let $(X_i,Y_i)_{i=1,...,n}$ be an i.i.d. sample
drawn from the distribution $N(\mu,C)$ with unknown $\mu$ and
$C$. Then $\rfieller$  as given in  Definition 
\ref{construction-fieller} is an  exact confidence region of level
$1-\alpha$ for $\rho$. 
\end{theorem}
{\em Proof of Fieller's theorem (sketch).} 
Consider the function
\banum
\label{eq-statistic-T}
T_{r,\hat C}(x) 
:= 
\frac{x_2- r x_1}{\sqrt{\hat{c_{22}}
    - 2 {r} \hat{c_{12}} + {r^2} \hat{c_{11}}} }
\eanum
where $r \in \RR$ is a parameter and $\hat C$ denotes the sample
covariance matrix. 
If applied to $r =
\rho$ and $x = \hat\mu$, the statistic $T_{\rho,\hat C}(\hat \mu)$ has
a Student-t distribution with $(n-1)$ degrees of freedom. %
The set $\rfieller:= \{r \in \RR |\; T_{r,\hat C}(\hat\mu) \in [-q,q]
\}$ now satisfies (by the definition of $q$ as Student-t quantile)  %
\ba
P(\rho \in \rfieller) 
= P(T_{\rho,\hat C}(\hat\mu) \in  [-q,q] )
= 1 - \alpha.
\ea
Solving $-q \leq T_{r,\hat C}(\hat \mu) \leq q$ for $r$  leads to a 
quadratic inequality whose solutions are given by Fieller's theorem.  \ulesqed \\

Let us make a few comments about this proof. The most important
property of the statistic $T_{\rho,\hat C}(\hat\mu)$ is the fact that
its distribution does not depend on $\rho$. That is, it is a pivotal
quantity. Otherwise, solving the inequalities $-q \leq T_{r,\hat
  C}(\hat \mu) \leq q$ for $r$ would not lead to an expression which
is independent of $\rho$. Moreover, note that the mapping
$T_{\rho,\hat C}$ projects the points on the line $L_\rhoperp$, and
additionally scales them such that the projected sample mean has
variance 1. In particular it is interesting to note that because
$T_{\rho,\hat C}(\mu) = 0$, the set
$
J_\rho = %
[T_{\rho,\hat C}(\hat\mu) - q, T_{\rho,\hat C}(\hat\mu) +q]
$
is a $(1-\alpha)$ confidence interval for the projected mean
$T_{\rho,\hat C}(\mu)$: 
\ba  
P(T_{\rho,\hat C}(\mu) \in J_\rho) 
& = 
P(0 \in [T_{\rho,\hat C}(\hat\mu) - q, T_{\rho,\hat C}(\hat\mu) +q])
& = P(T_{\rho,\hat C}(\hat\mu) \in [-q,q]) 
= 1 - \alpha.
\ea
This property will be used later on to generalize Fieller's confidence
set to more general distributions. Also note that  solving the
inequality  $-q \leq T_{r,\hat C}(\hat \mu) \leq q$ 
coincides with the
construction of the wedge in the geometric construction. The wedge can
be seen as exactly the lines with slope $r$ such that the projection of $\hat\mu$ on
$L_{\hat \rperp}$ is still within $[-q,q]$.\\

Based on all those observations it is 
very natural to expect a close relation between $\rfieller$ and $\rgeo$. Still, a
priori it is not clear that those two confidence sets coincide, as
confidence sets are not necessarily unique. But the following theorem
proves that this is indeed the case:

\begin{theorem}[$\rgeo$ and $\rfieller$ coincide]\sloppy
\label{th-we-equal-fieller}
The confidence region $\rgeo$ defined in Construction
\ref{construction-fieller-geometric} coincides with $\rfieller$ as given in Definition  \ref{construction-fieller}.
\end{theorem}

\begin{proof} {\em (Sketch)} First one has to show that the three
  cases in Fieller's theorem coincide with the three cases in the
  geometric approach. Second, one then has to verify that the numbers
  $l_1$ and $l_2$ in Fieller's theorem coincide with the slopes of the
  tangents to the ellipse. Both steps can be solved by straightforward
  but lengthy calculations. Details can be found in \citet{LuxFra04}.
\end{proof}
Note that in the proof of Fieller's theorem we did not directly use
the fact that we have {\em paired}
samples $(X_i,Y_i)_{i=1,...,n}$. Indeed, Fieller's theorem and its
proof can also be valid in the more general setting where we
are given two independent samples  $X_1,...,X_n$ %
and $Y_1,...,Y_m$ with a different number of sample points, and use
unbiased estimators for the means $\mu_1$, $\mu_2$ and independent
unbiased estimators for the (co)variances $\hat{c_{ij}}$. In this case
one has to take care to choose the degrees of freedom
in the Student-$t$-distribution appropriately, see
\citet{Buonaccorsi_01} and Section 3.3.3 of \citet{Rencher98}.\\
\sectionsinica{Exact confidence sets for general random
  variables}{sec-exact-non-normally} 
In this section we show how to extend our geometric approach to
non-normally distributed random variables. While it is straightforward
to extend our geometric approach to elliptically symmetric
distributions, re-interpreting the geometric construction also leads
to a new construction for more general circumstances. \\

 \subsectionsinica{Elliptically symmetric distributions} 
 In the normally distributed case, the main reason why Construction
 \ref{construction-fieller-geometric} leads to exact confidence sets
 is that the projected and studentized mean is Student-$t$
 distributed, no matter in which direction we project. More generally,
 such a property holds for all elliptically symmetric random
 variables.
Elliptically symmetric random variables can be written in the form
$\mu + A Y$ where $\mu$ is a shift parameter, $A$ is a matrix with
$AA' = C$, and $Y$ any spherically symmetric random variable generated
by some distribution $H$ on $\RR_{+}$. For a brief overview of
spherical and elliptical distributions see \citet{Eaton81}, for an
extensive treatment see \citet{FanKotNg90}. In particular, if $X$ is
an elliptically symmetric random variable with shift $\mu$, covariance
$C$, and generator $H$, then the statistic $T_{r,\hat C}(\hat\mu)$
introduced in Equation \eqref{eq-statistic-T} is a pivotal quantity
which has the same distribution for all $r \in \RR$.  Denote the
distribution function of this statistic by $G$.  To extend
Construction~\ref{construction-fieller-geometric} to the case of
elliptically symmetric distributions, all we have to do is to define
the quantile $q$ in Construction \ref{construction-fieller-geometric}
or Definition \ref{construction-fieller} by the quantile
$q(G,1-{\alpha}/{2})$ of the distribution $G$. With similar arguments
as in the last sections one can see that the resulting confidence set is
exact. \\

\subsectionsinica{Confidence sets for a very general class of distributions}
Once we leave the class of elliptically symmetric distributions, the
distributions of the projected means are no longer independent of the
direction of the projection, and all the techniques presented above
cannot be used any more.  However, there is a surprisingly simple way
to circumvent this problem. To see this, let us re-interpret
Construction \ref{construction-fieller-geometric} as depicted in
Figure \ref{fig-fieller-generalization-proof}.
Previously, to determine whether $r \in \RR$ should be element of
$R_{geo}$  we checked whether the line with slope $r$ is
inside the wedge enclosing the ellipse $E$. But note that the
same result can be achieved if we project the sample on the line
$L_\rperp$, construct a one-dimensional confidence set $J_r$ for the
mean on $L_\rperp$, and check whether $0 \in J_r$ or not. This
observation is the key to the following construction:

\begin{figure}[!bt]
  \begin{center}
 \includegraphics[width=0.45\textwidth]{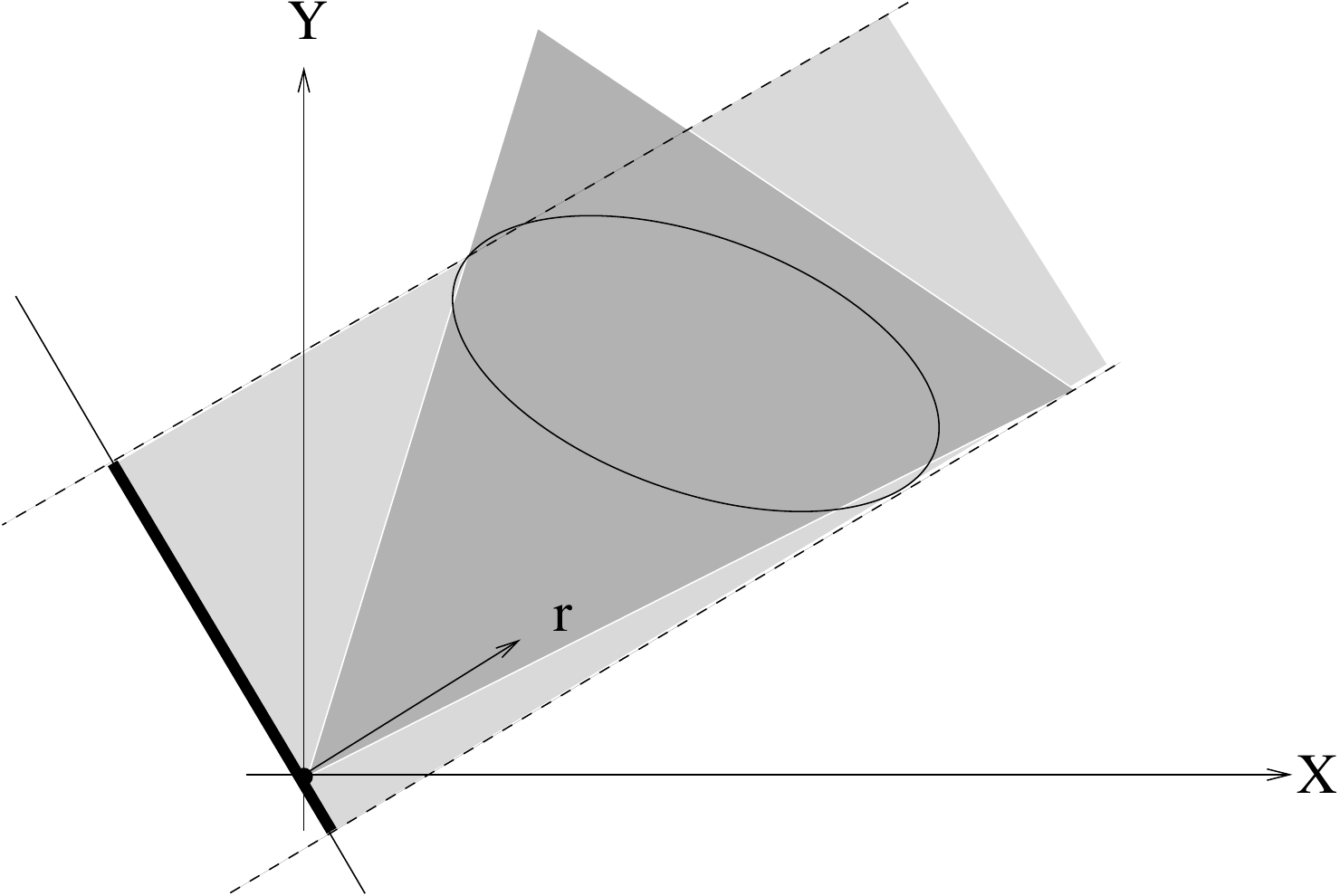}
\includegraphics[width=0.45\textwidth]{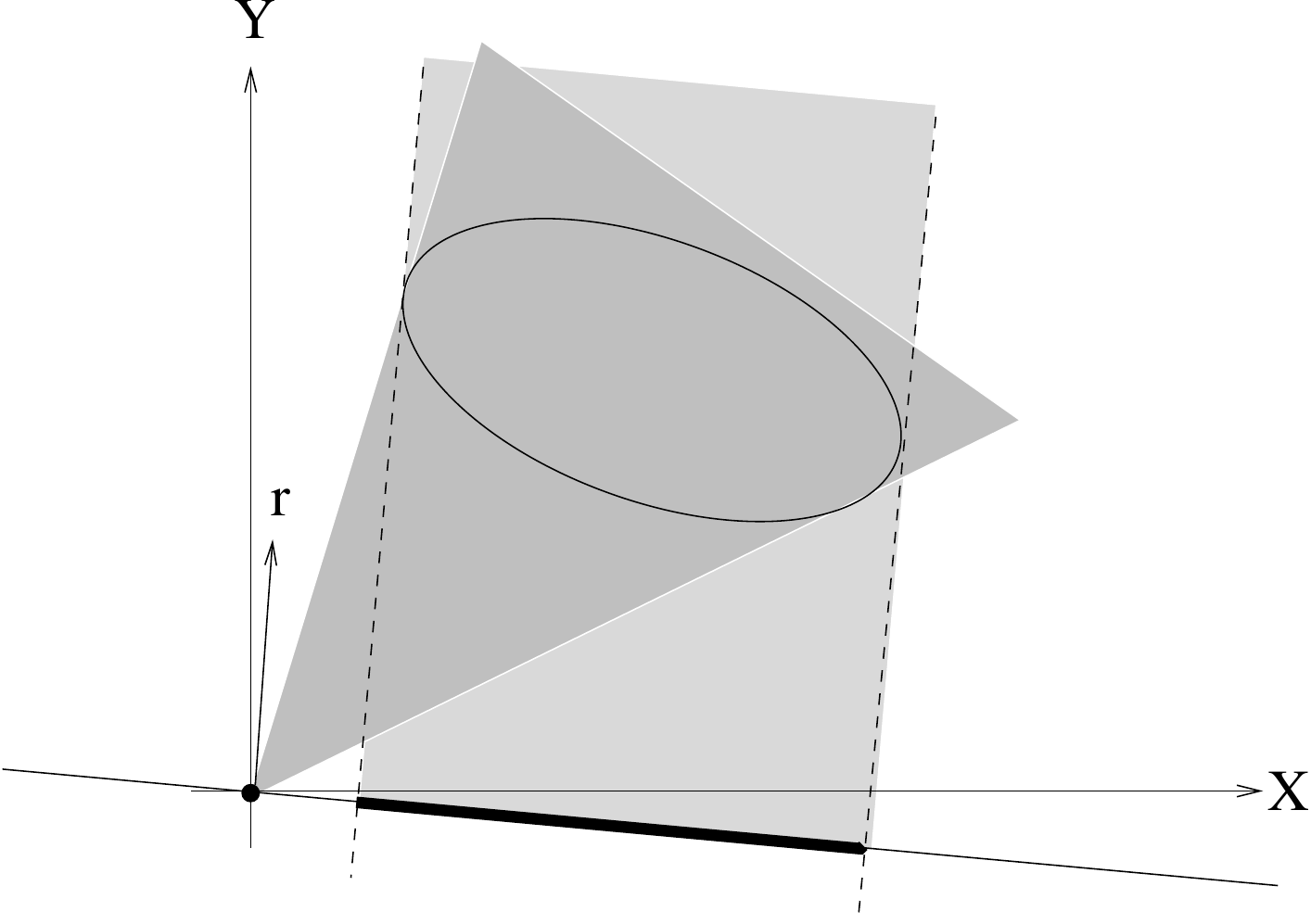}
\caption{Second geometric interpretation: By definition, ratio $r$ is element of
  Fieller's confidence set $\rgeo$ if the line $L_r$ (depicted by the
  little arrow) is inside the wedge enclosing the covariance ellipse. This is the case if
  and only if the origin is inside the projection 
$J_r := \pi_\rperp(E)$ of the ellipse on the line
  $L_\rperp$. The left panel shows a case where $r
  \in \rgeo$, the right panel a case where $r \not\in \rgeo$.  } 
  \label{fig-fieller-generalization-proof}
  \end{center}
\end{figure}
\begin{construction} \label{construction-exact-general}
{\bf (Exact confidence sets $\boldsymbol{\rgen}$ for $\boldsymbol\rho$ in case of general distributions)}
\begin{enumerate}
\item For each $r \in \R$, project the sample points on $L_\rperp$, that is define the new points
  $U_{r,i} = \pi_\rperp(X_i,Y_i)$, $i=1,\hdots,n$. 
\item For each $r \in \R$, construct a confidence set $J_r$ for the mean of
  $U_{r,i}$, that is  a set such that $P(\pi_\rperp(\mu) \in J_r) = 1 - \alpha$. 
\item Then define the
  confidence set $\rgen$ for $\rho$ as $\rgen = \{r \in \R \; | \; 0 \in
    J_r\}$.
\end{enumerate} 
\end{construction}

The big advantage of this
construction is that the projection in direction of the true value
$\rho$ is not singled out as a ``special'' projection, we simply look
at {\em all projections}. Hence, Construction
\ref{construction-exact-general} does not require any knowledge about
$\rho$.  
\begin{theorem}[$\rgen$ is an exact confidence set for $\rho$]
\label{theorem-generalization-is-exact}
Let $(X_i, Y_i)_{i=1,...,n} \in \RR^2 $ be i.i.d. pairs of random
variables with arbitrary distribution such that the joint mean of
$(X,Y)$ exists. If the confidence sets $J_r$ used in Construction
\ref{construction-exact-general} exist and are exact (resp. conservative resp.
liberal) confidence sets of level $(1-\alpha)$ for the means of
$\pi_\rperp((X_i,Y_i))_{i=1,...,n}$, then $\rgen$ is an exact (resp.
conservative resp. liberal) confidence set for $\rho$.
\end{theorem}
\begin{proof} In the exact case, we have to prove that the true ratio $\rho$ satisfies
  $P(\rho \in \rgen) = (1-\alpha)$. By definition of $\rgen$, for each
  $r\in\RR$ we have that $r \in \rgen \iff 0 \in J_r$. In particular, this
  also holds for $r = \rho$. Moreover, the projection corresponding to
  the true ratio $\rho$ projects the true mean $\mu$ on the
  origin of the coordinate system. By linearity, the projection of the true
  mean $\pi_\rhoperp(\mu)$ equals the mean of the projected
  random variables. By construction of $J_r$ we know
  that the latter is inside $J_r$ with probability exactly
  $(1-\alpha)$. So we can conclude that 
$
P(\rho \in \rgen) = 
P(0 \in J_\rho) = 
P(\pi_\rhoperp( \mu) \in J_\rho) = 
1-\alpha.
$
\end{proof}

We proved that the set $\rgen$ defined in Theorem
\ref{theorem-generalization-is-exact} is an
{\em exact} confidence set for the ratio of random variables.  The
only assumptions are that the means of $X$ and $Y$ exist and that
there is a rule to compute exact confidence intervals for the means of
the projections $\pi_\rperp(X,Y)$. 
To our knowledge, Construction \ref{construction-exact-general} is the
first construction of exact confidence sets for general distributions.
It reduces the difficult problem of estimating confidence sets for the
ratio of two random variables to the problem of estimating confidence
sets for the means of one-dimensional random variables.  On a first
glance this looks very promising.
However, the crux for applying this construction in practice is that
one has to know the analytic form of the distribution of the projected
means. For this one has to be able to derive an analytic expression
for general linear combinations of $X$ and $Y$. While there might be
some special cases in which this is tractable, for the vast majority of
distributions such an analytic form is not easy to obtain. As a
consequence, while being of theoretic interest, Construction
\ref{construction-exact-general} is of limited relevance for practical
applications.\\

\sectionsinica{Conservative confidence sets for more general random variables}{sec-conservative}
Our geometric principles can also be used to derive very simple
conservative confidence sets for general distributions. The main idea
is to replace the ellipse used in Construction \ref{construction-fieller-geometric}
by a more general convex set $M \subset \R^2$. A straightforward idea
is choose $M$ as a $(1-\alpha)$-confidence set for the bivariate
joint mean $\mu \in \RR^2$, that is a set such that $P(\mu \in M) =
1-\alpha$. Then, as above we can construct the wedge $W$  around $M$ which is
given by the two enclosing tangents and choose a confidence region
$\rcons$ by intersecting the wedge with the line $x=1$, distinguishing
between the same three cases as above. %
For general distributions, there exists a simple but
effective way to choose the convex set $M$. Namely, we take the axis-parallel
rectangle $A := I_1 \times I_2$, 
where the intervals
$I_1:=[l_1,u_1]$ and $I_2:=[l_2,u_2]$ are confidence
intervals for the {one-dimensional} means $\mu_1$ of $X$ and $\mu_2$ of
$Y$. 
Formally, the construction is the following: 
\begin{construction} \label{construction-conservative}
{\bf (Geometric construction of conservative
  confidence regions $\rcons$ for $\boldsymbol{\rho}$
  for general distributions)}
\begin{enumerate}
\item Construct exact confidence intervals $I_1$  and $I_2$ of level
  $(1-\alpha/2)$ for the means of
  $X$ and  $Y$, respectively. In the two-dimensional plane,
  define the rectangle $A=I_1
  \times I_2$. 
\item
\begin{enumerate}
\item If $(0,0) \text{ is not inside } A$, construct the two tangents to
  $A$ through the origin $(0,0)$, and let $W$ be the wedge enclosed by
  those tangents.  Define the confidence
  region $\rcons$ as the intersection of $W$ with the line
  $x=1$. Depending on whether the $y$-axis lies inside $W$ or not this
  results in an exclusive unbounded or a bounded confidence region
\item  If $(0,0) \text{ inside } A$, choose the confidence region as  $\rcons=
       ]-\infty,\infty[$. \\
\end{enumerate}
\end{enumerate}
\end{construction} 
\begin{theorem}[$\rcons$ is a conservative confidence set for $\rho$]
\label{theorem-conservative} \sloppy
Let $(X_i, Y_i)_{i=1,...,n} \in \RR^2 $ be i.i.d. pairs of random
variables with arbitrary distribution such that the joint mean of
$(X,Y)$ exists. If the confidence sets $I_1$ and $I_2$ used in
Construction \ref{construction-conservative} exist and are exact or
conservative confidence sets of level $(1-\alpha)$ for the means of
$X$ and $Y$, then $\rcons$ is a conservative confidence set for $\rho$
of level $(1-2\alpha)$.
\end{theorem}
The proof of this theorem is nearly trivial and can be given in two lines:
\ba
P(\rho & \in  \rcons) = P(\mu \in W) 
 \geq
P(\mu \in A) 
=
P(\mu_1 \in I_1 \text{ and } \mu_2 \in I_2) \\
& = 1 - P(\mu_1 \not\in I_1 \text{ or } \mu_2 \not\in I_2) 
\geq
1 - ( P(\mu_1 \not\in I_1) + P(\mu_2 \not\in I_2)) 
 = 1 -2\alpha. %
\hspace{0.5cm}\ulesqed
\ea
Interestingly, it can be seen 
easily that the set $\rcons$ constructed using the rectangle coincides with the set obtained by
``dividing'' the one-dimensional confidence intervals $I_2$ by $I_1$, namely
$
\rcons = I_2/I_1 := \left\{\frac{y}{x}; \;  y \in I_2,x\in I_1\right\}.
$
The latter is a heuristic for confidence sets for ratios which can
sometimes be found in the literature, usually without any theoretical
justification. Our geometric method now reveals effortlessly that it is
statistically safe to use this heuristic, but that it will lead to
conservative confidence sets of level $1 - 2 \alpha$.

Of course, one could think of even more general ways to construct a
convex set $M \subset \R^2$ as base for the conservative geometric
construction. For example, instead of using axis-parallel projections
as in Construction \ref{construction-conservative}, one could base the
convex set $M$ on projections in arbitrary directions (for example,
using the two projections in direction of $\rho$ and $\rhoperp$, or
even using more than two projections). However, we would like to
stress one big advantage of using the axis-parallel rectangle.  While
the exact generalizations presented in Section
\ref{sec-exact-non-normally} require to construct confidence sets for
the means of arbitrary linear combinations of the form $a X + b Y$,
for the rectangle construction we only need to be able to construct
exact confidence sets for the marginal distributions of $X$ and $Y$,
respectively. One can envisage many situations where distributional
assumptions on $X$ and $Y$ are reasonable, but where the distributions
of projections of the form $aX + bY$ cannot be computed in closed
form. In such a situation, the rectangle construction can serve as an
easy loophole. The prize we pay is the one of obtaining conservative
confidence sets for the ratio instead of exact ones. But in many
cases, obtaining confidence sets which are provably conservative might
be preferred over using heuristics with unknown guarantees to
approximate exact confidence
sets.\\

\sectionsinica{Bootstrap confidence sets}{sec-bootstrap}
In the last sections we have seen how exact and conservative
confidence sets for ratios of very general classes of distributions
can be constructed. In practice, the application of those methods is
limited by the problem that we still need strong assumptions to apply
them:  we need to know the exact distributions of the
projections of $(X,Y)$.  In this section we want to investigate how
approximate confidence sets can be constructed in cases where the
underlying distributions are unknown. 
A natural candidate to construct approximate confidence sets for
ratios are bootstrap procedures (e.g., \citealp{Efron_79,Efron_Tibshirani_93,ShaTu95,DavHin97}).
However, if the variance of the statistics of interest does not exist,
as is usually the case for $\hat\rho$, bootstrap confidence regions
can be erroneous \citep{Athreya_87,Knight_89}. Moreover, standard
bootstrap methods which attempt to bootstrap the statistic $\hat\rho$
directly cannot result in unbounded confidence regions. This is
problematic, as it has been shown that any method which is not able to
generate unbounded confidence limits for a ratio can lead to arbitrary
large deviations from the intended confidence level
\citep{Gleser_Hwang_87,Koschat_87,Hwang_95}.  Hence, bootstrapping
$\hat\rho$ directly is not an option. 
Instead, in the
literature there are several approaches to use bootstrap methods based
on the studentized statistic $T_{\rho,\hat C}(\hat \mu)$ introduced in Equation
\eqref{eq-statistic-T}.  A simple approach along those lines is taken in
\citet{ChoLecLeg99}. The authors use standard bootstrap methods to construct a
confidence interval $[q1, q2]$ for the mean of the statistic
$T_{\hat\rho,\hat C}(\hat\mu)$. As confidence set for the ratio, they
then use the interval $[\hat\rho - q_2 S_{\hat\rho}, \hat\rho
-q_1 S_{\hat\rho}]$ where $S_{\hat\rho}$ is the estimated standard
deviation of $\hat\rho$.  However, this approach
is problematic: the confidence sets do not have the qualitative
behavior as the Fieller ones, and as they are always finite, the
coverage probability can be arbitrarily small. \\

\subsectionsinica{Bootstrap approach by Hwang and its
  geometric interpretation}
A more promising bootstrap approach for ratios has been presented by
\citet{Hwang_95}. He suggests to use standard bootstrap methods to construct 
  confidence sets for the mean of $T_{\hat\rho,\hat C}(\hat\mu)$. To determine the
  confidence set for the ratio, he then proceeds as Fieller
  and solves a quadratic equation to determine the confidence set
  for the ratio.  %
\citet{Hwang_95} argues that his confidence sets are advantageous 
when dealing with asymmetric distributions such as
  exponential distributions. %
  However, we need to be careful here. \citet{Hwang_95} only 
treats the case of one-sided confidence sets,
  where he 
   constructs a confidence set of the form $]-\infty,q]$ for
  $T_{\hat\rho,\hat C}(\hat\mu)$ and then solves the quadratic
  equation $T_{\hat\rho,\hat C}(\hat\mu)^2 \leq q^2$. This leads to
  the three well-known cases bounded, exclusively unbounded,
  completely unbounded. However, the two-sided case is more involved
  and is not discussed in his paper. If one uses
  symmetric bootstrap confidence sets of the form $[-q,q]$ for
  $T_{\hat\rho,\hat C}(\hat\mu)$, then one can proceed by solving one
  quadratic inequality similar to above.  However,
  if one wants to exploit the fact that the distribution might not be
  symmetric, one would have to use asymmetric (for example equal-tailed) confidence sets of the
  form $[q1,q2]$ for $T_{\hat\rho,\hat C}(\hat\mu)$. But then, solving
  the equations $ q1 \leq T_{\hat\rho,\hat C}(\hat\mu) \leq q_2$ can
  lead to unpleasant effects. To satisfy both inequalities
  simultaneously, one has to solve two different quadratic
  inequalities. The joint solution can not only
  attain the three Fieller types, but all possible intersections of
  two Fieller type sets. For example, one can obtain confidence sets for the
  ratio which are only unbounded on one side, such as %
  $]-\infty, l] \cup [l', u]$. Such confidence sets are quite 
  implausible: as we discussed after Construction
  \ref{construction-fieller-geometric}, in cases where the denominator
  is not significantly different from 0 the confidence set should be
  unbounded on both ends. Otherwise, the confidence set of the ratio
  would reflect a certainty about the sign of the denominator that is
  not present in the confidence set of the denominator itself. 
Consequently, we believe that Hwang's
  approach should only be used with symmetric (and not with
  equal-tailed) confidence sets for $T_{\hat\rho,\hat C}(\hat\mu)$.
  In this case, Hwang's bootstrap approach can easily be interpreted
  in our geometric approach and is in fact very similar to Fieller's
  approach: as in Construction \ref{construction-fieller-geometric},
  one forms the covariance ellipse centered at $\hat\mu$ using the
  estimated covariance matrix $\hat C$. But instead of using quantiles
  of the Student-t distribution to determine the width $q$ of the
  ellipse, one now uses  bootstrap quantiles for this purpose. Then one
  proceeds exactly as in the Fieller case. This geometric
  interpretation  reveals that Hwang's approach relies on one
  crucial assumption on the distribution of the sample means: their
  covariance structure has to be elliptical. So while seeming 
  distribution-free at first glance, Hwang's bootstrap
  approach with symmetric confidence sets relies on the implicit assumption that the sample mean is elliptically
  distributed. Below we will illustrate some consequences of this
  insight in simulations.\\

\subsectionsinica{A geometric bootstrap approach}
We now want to suggest a bootstrap approach which potentially is more
suited to deal with highly asymmetric distributions. To this end, we
will adapt the geometric Construction \ref{construction-conservative}
to a bootstrap setting. This can be done in a straightforward manner: we simply use
bootstrap methods to construct the one-dimensional confidence
intervals $I_1$ and $I_2$ used in Construction
\ref{construction-conservative}, and then proceed exactly as in
Construction \ref{construction-conservative}. The advantage of this
approach is obvious: we do not need to make any assumptions on the
distribution, can easily use asymmetric confidence intervals $I_1$ and
$I_2$, and still obtain a Fieller-type behavior (as opposed to Hwang's
method, which does not have this behavior when using asymmetric
bootstrap sets).  Moreover, our construction does not assume
elliptical covariance structure, and can, for example, be used for
heavy-tailed distributions which are not in the domain of attraction
of the normal law. In this sense, the geometric bootstrap approach can
be applied in situations where both Fieller's and Hwang's
confidence sets fail. This will be demonstrated below.  \\

Note that one can easily come up with other, more involved
bootstrap methods based on the geometric method. For example, one can
use more than two projections, one can use projections which are not
parallel to the coordinate axes, or one can even base the wedge on
more general two-dimensional convex sets in the plane. A completely
different approach can be based on bootstrapping polar representations of the data
(along the lines of \citealp{Koschat_87}). However, given that in
our simulations those methods did not perform better than the existing
methods we will not discuss those approaches in detail. \\

\subsectionsinica{Simulation study} 
In this section we would like to present some numerical simulations to
compare the bootstrap approach by Hwang, our geometric bootstrap
approach, and Fieller's standard confidence set. \\
{\bf Setup. } 
For both $X$ and $Y$ we use three different types
of distributions: \\
{\em Normal distributions. } Here we always fixed the mean
to 1 and varied the variance between 0.1 and 10.\\
 {\em Exponential distributions.} They are highly asymmetric, but
  still in the domain of attraction of the normal law. Here we varied the
  mean between 0.1 and 10. \\
{\em Pareto distributions } with density function $p(x) = a k^a /
  x^{a+1}$, cf. Chapter 20 of \citet{JohKotBal94}.  For a Pareto(k,a)
  distributed random variable, all moments of order larger than $a$ exist, the
  smaller moments do not exist. In particular, for $a \in
  ]1,2[$, the expectation exists, but the variance does not
  exist. In this case the
  distribution is heavy-tailed and not in the domain of
  attraction of the normal law. In our experiments, we varied the tail parameter $a$
  between 1.1 and 2.5 and always chose parameter $k$ such that the
  expectation is 1 (that is, we chose $k=(a-1)/a$). For some
  simulations we also used an inverted Pareto distribution (a Pareto
  distribution which has been flipped around its mean, so that its
  tail goes in the negative direction).\\
For each fixed distribution of $X$ and $Y$, we independently sampled
$n=20$ ($n=100$, $n=1000$, respectively) data points $X_i$ and
$Y_i$. Then we computed 
the Fieller confidence set according to Definition
\ref{construction-fieller}, our geometric bootstrap confidence sets as
introduced above, and Hwang's bootstrap confidence sets. Each simulation was
repeated $R=1000$ times to compute the empirical coverage. As nominal
coverage probability we always chose $90\%$ (in terms of coverage, this is
more meaningful than the level $95\%$ as it leaves more room for
deviations in both directions). 
To construct the bootstrap confidence sets for the one-dimensional
means of $X$ and $Y$ (in the geometric method) and the projection
$T_{\hat\rho,\hat C}(\hat\mu)$ (in Hwang's method) we used different
bootstrap methods.  As default bootstrap method we used bootstrap-t
(cf.  \citealp{Efron_Tibshirani_93}). We also tried several other
standard methods such as the percentile or the bias corrected and
accelerated (BCA) method (cf.  \citealp{Efron_Tibshirani_93}), but did
not observe qualitatively different behavior.  To deal with
heavy-tailed distributions, we applied methods based on subsampling
self-normalizing sums, as introduced by \citet{HalLep96}, see also
\citet{RomWol99}.  Here one has to choose one parameter, namely the
size $m$ of the subsamples. We did not use any automatic method to
optimize this parameter, but based on values reported in
\citet{RomWol99} we fixed it to $m=10$ ($40, 400$) for $n = 20$ ($100,
1000$).  For all bootstrap methods, we tried both equal-tailed and
symmetric confidence sets, in all cases with $B=2000$ bootstrap
samples.  We will report the bootstrap results using notations such as
Hwang(symmetric, bootstrap-t) or Geometric(equal-tailed, Hall). The
terms in parentheses always refer to the
construction of the confidence sets for the respective one-dimensional projections. \\

\begin{figure}[!bt]
  \begin{center}
 \includegraphics[width=\textwidth]{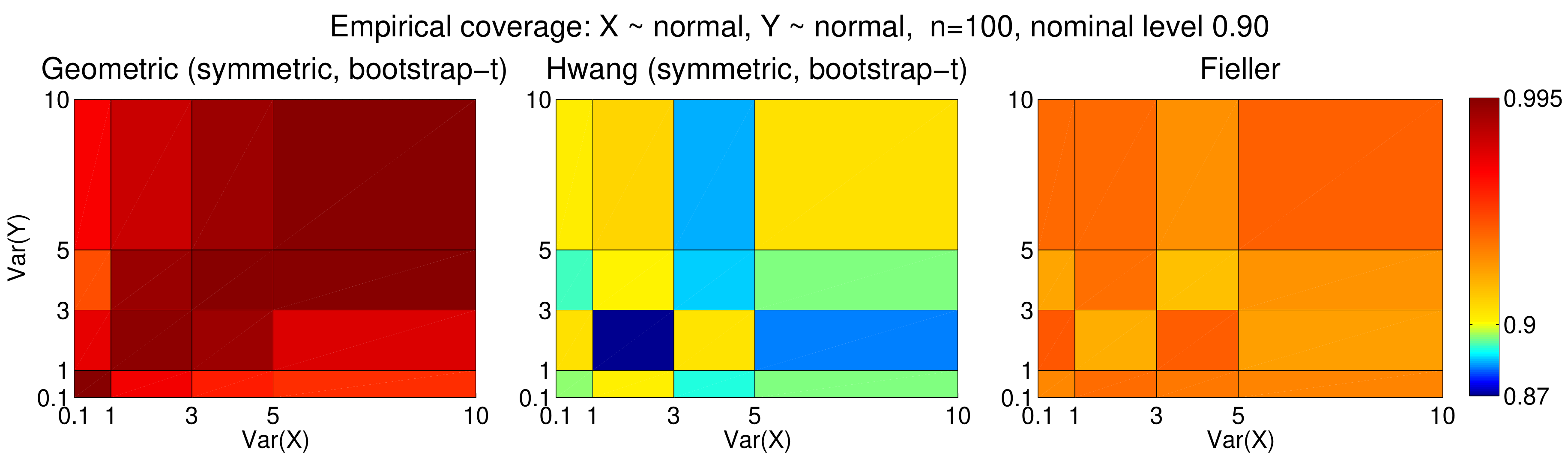}
 \includegraphics[width=\textwidth]{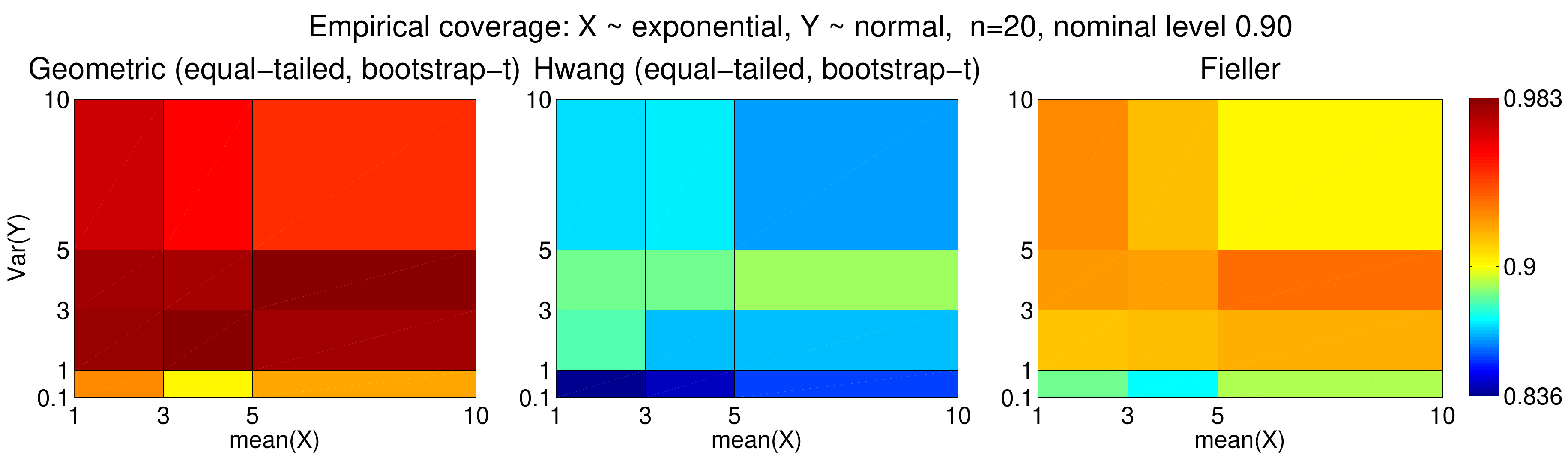}
 \includegraphics[width=\textwidth]{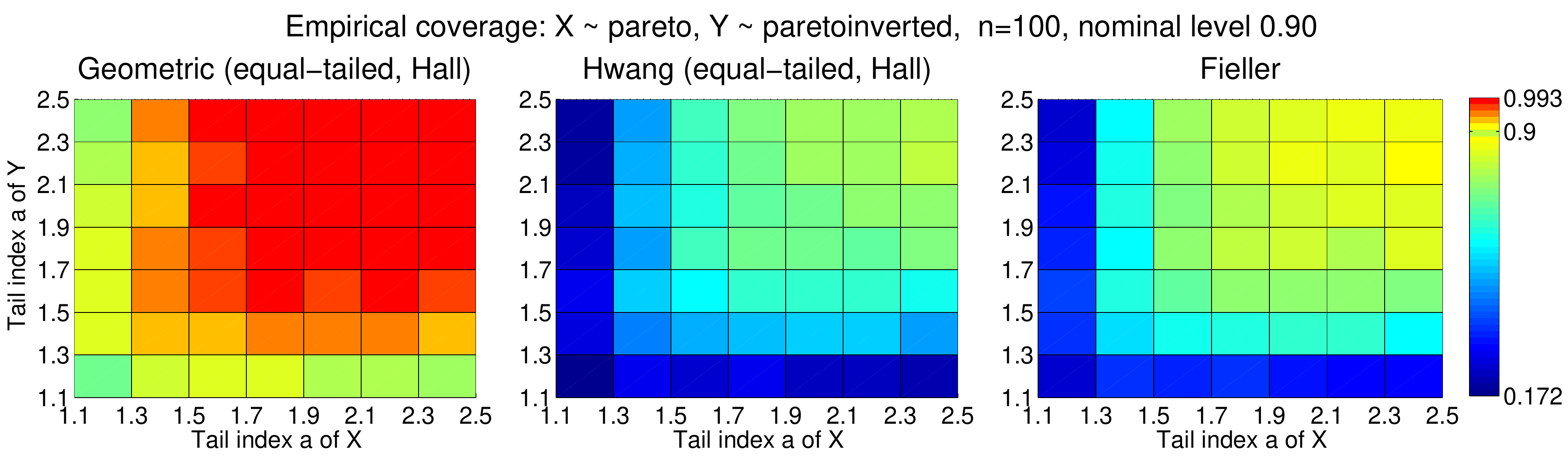}

\includegraphics[width=\textwidth]{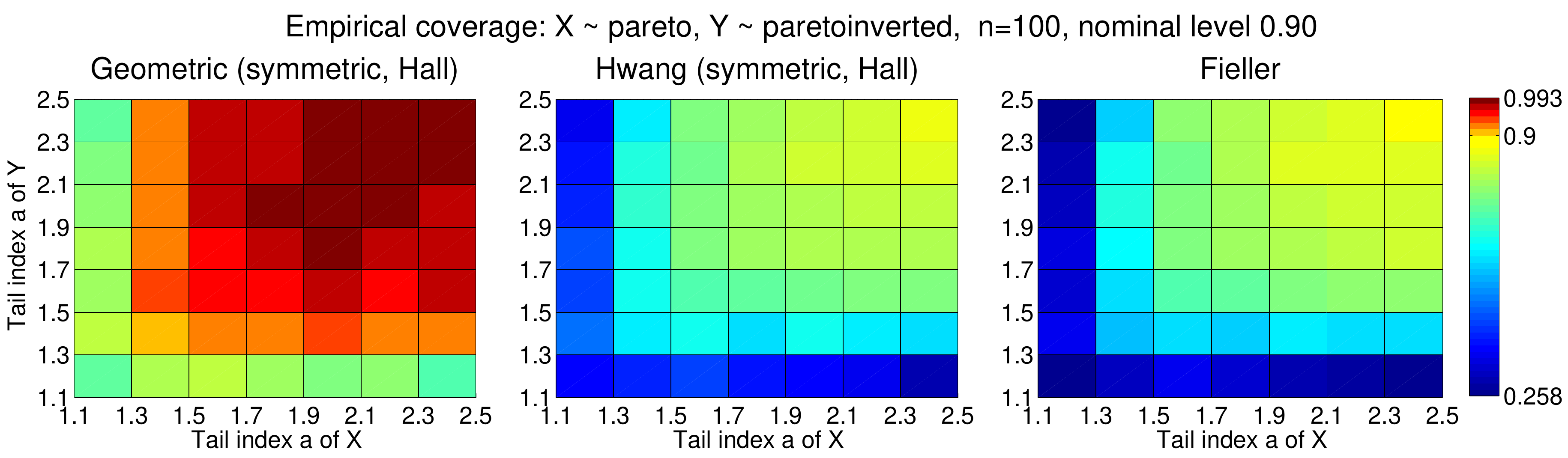}

 \includegraphics[width=\textwidth]{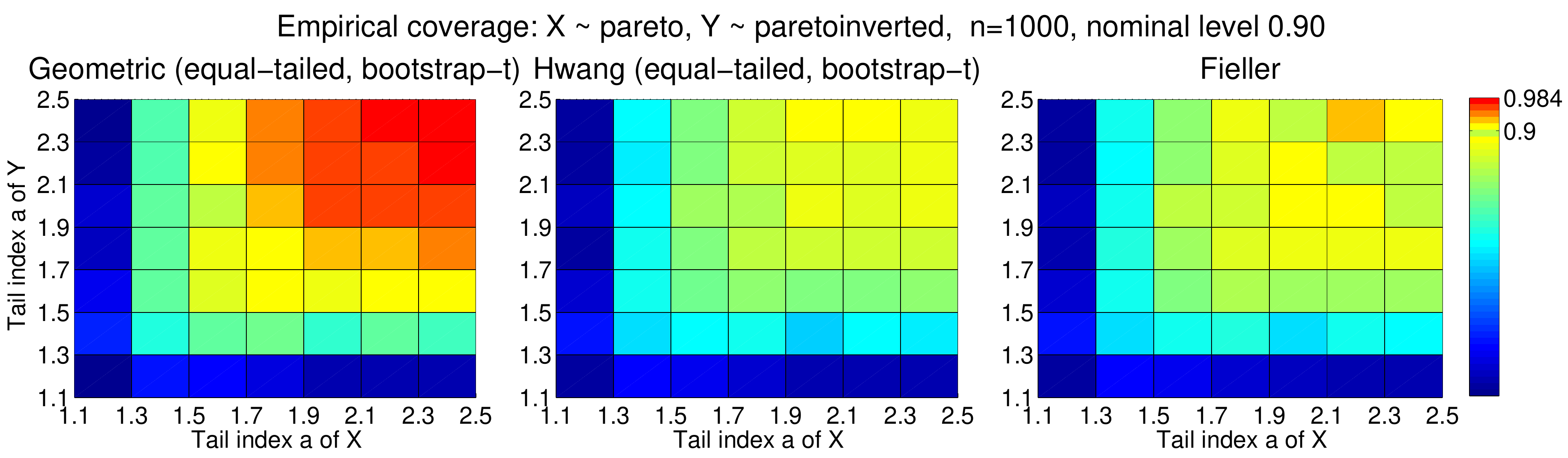}
 \caption{Empirical coverage. Each row corresponds to one
   fixed set of parameters, and shows the empirical coverage of the
   three methods. The nominal confidence level 0.90 is always
   depicted in yellow, red colors depict conservative and green/blue
   colors liberal confidence sets. The color scales are constant within each
   row, but change between the rows.  } 
 \label{fig-simulations-coverage}
\end{center}
\end{figure}
\begin{figure}[!bt]
  \begin{center}
 \includegraphics[width=\textwidth]{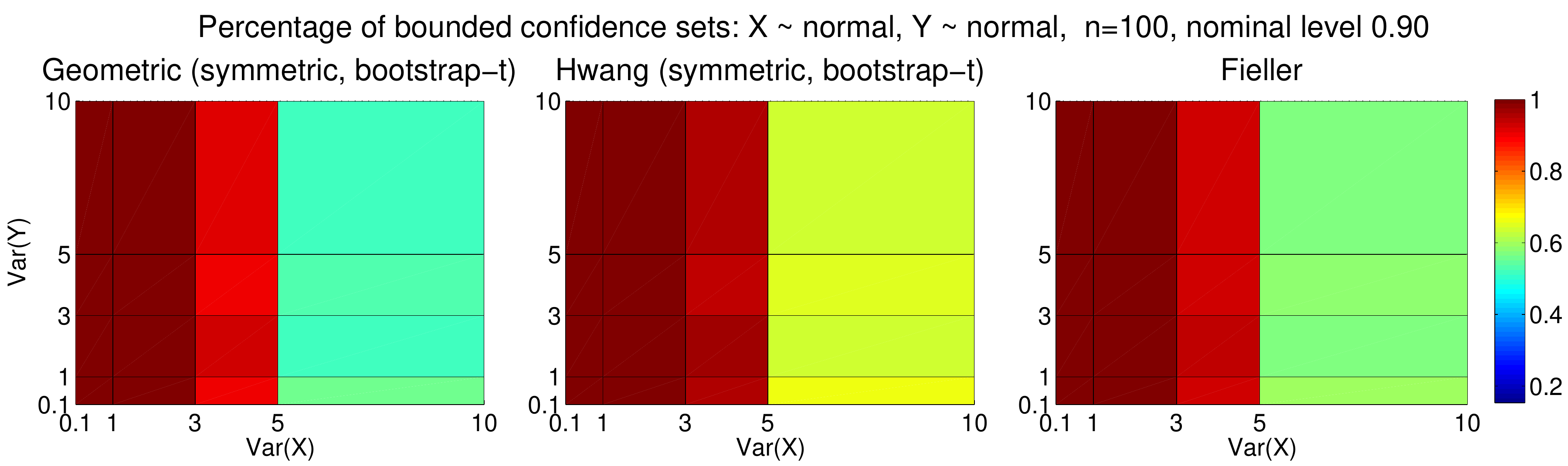}
 \includegraphics[width=\textwidth]{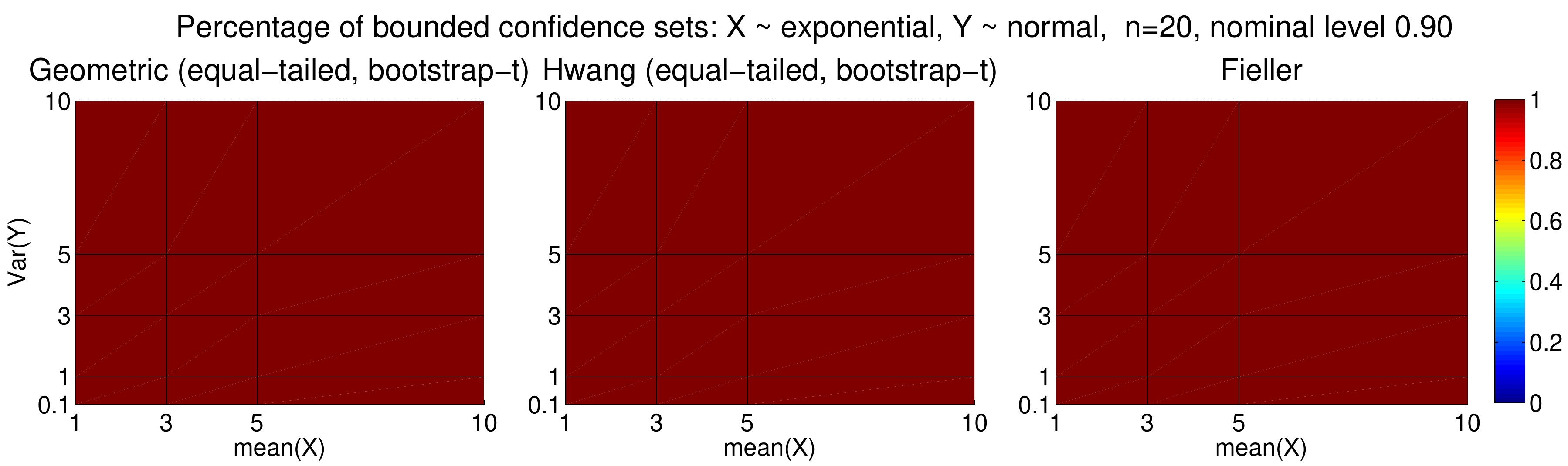}
 \includegraphics[width=\textwidth]{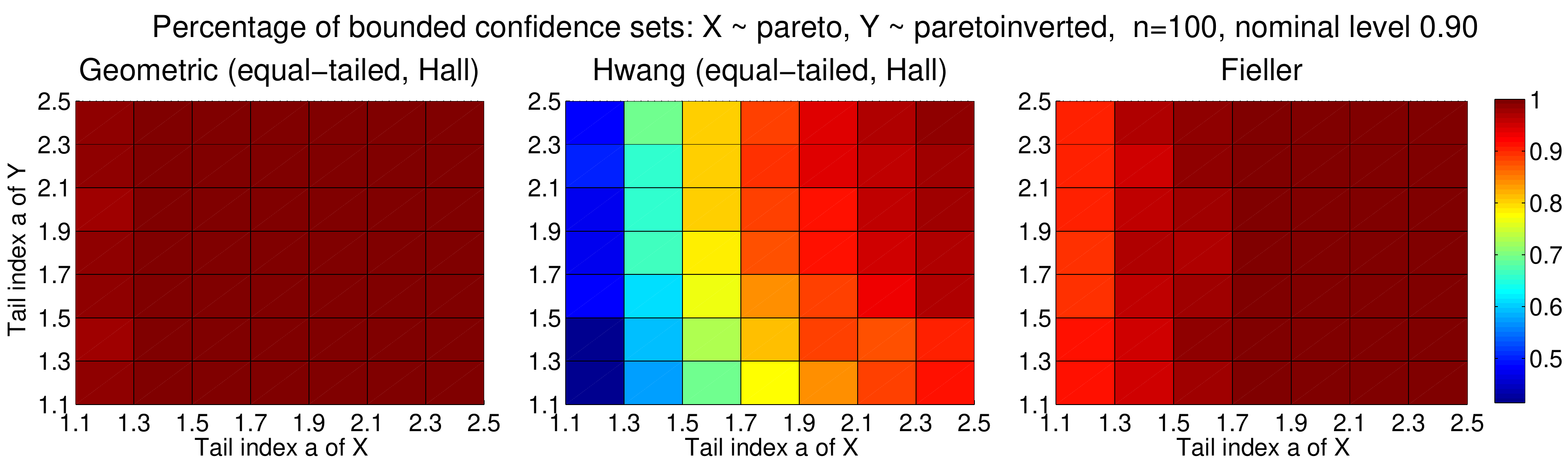} 
\includegraphics[width=\textwidth]{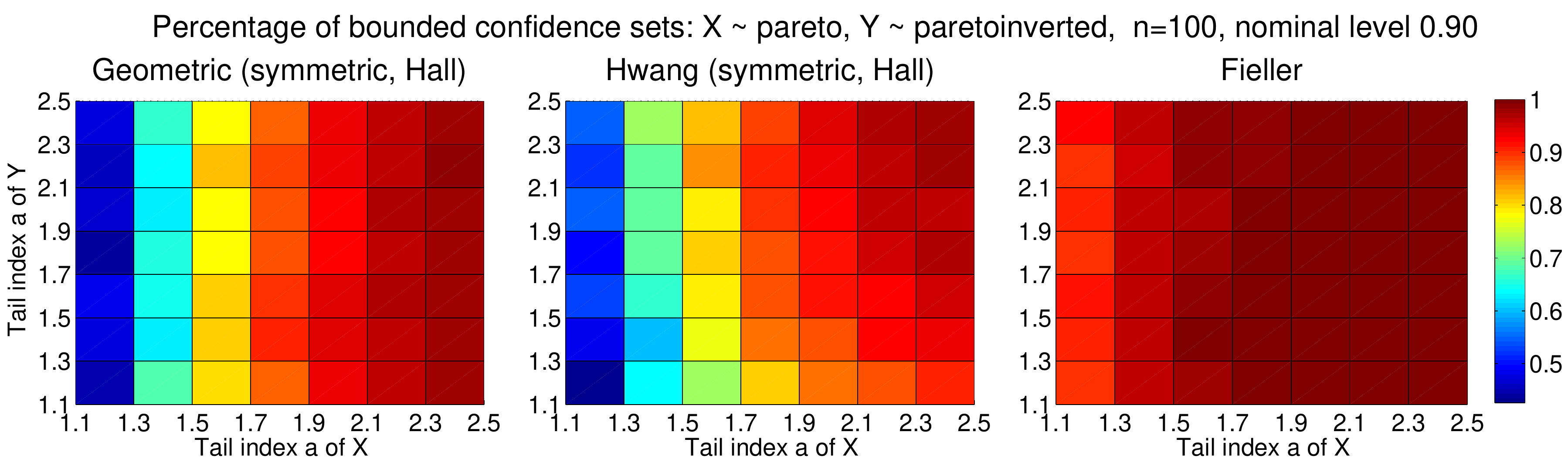}
 \includegraphics[width=\textwidth]{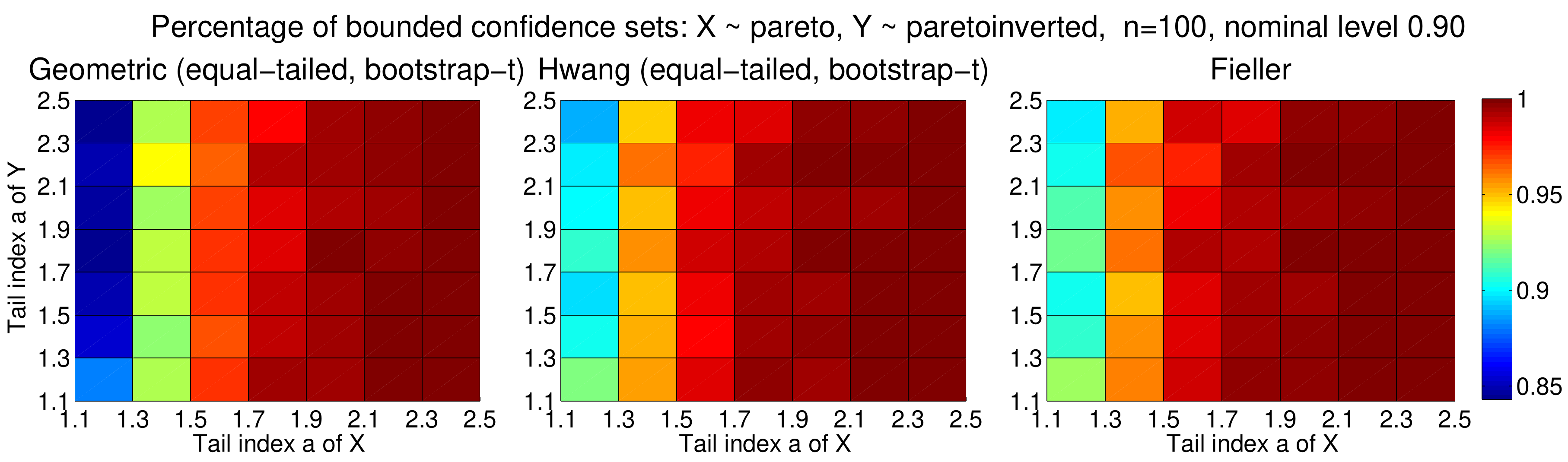}
 \caption{Percentage of bounded confidence sets (over 1000 simulations). Each
   row corresponds to one fixed set of parameters, and shows the
   percentage of bounded confidence sets for the three 
   methods. The color scales are constant within each
   row, but change between the rows. }
 \label{fig-simulations-bounded}
\end{center}
\end{figure}

{\bf Evaluation. } In all settings we evaluated the empirical coverage
(see Figure \ref{fig-simulations-coverage}) and the number of bounded
confidence sets (see Figure \ref{fig-simulations-bounded}). Due to
space constraints we cannot show the results for all parameter
settings in detail. Many more figures can be found in the
supplementary material to this paper \citep{LuxFra07_supplement}.

{\em Coverage properties in case of finite variance. }  We start with
the case where both $X$ and $Y$ are normally distributed
(Figure~\ref{fig-simulations-coverage}, first row).  Here Fieller's
confidence set is exact, and indeed we can see that it achieves very
good coverage values. In terms of absolute deviation from the nominal
confidence level, Hwang performs comparably to
Fieller. The difference is that Fieller tends to be slightly
conservative, while Hwang tends to be slightly liberal.%
As predicted, the geometric method is conservative and achieves higher
than nominal coverage. For all three methods, the results based on different
sample sizes and different bootstrap constructions are qualitatively
very similar (see supplement). 

To investigate the effect of symmetry, we consider the case where one
of the random variables is exponentially distributed and thus highly
asymmetric (Figure~\ref{fig-simulations-coverage}, second row).  We
can see that qualitatively, the three procedures behave as described
above (Fieller slightly conservative, Hwang slightly liberal,
geometric conservative), even for a small sample size $n=20$ (results
for larger $n$ are similar, see supplement). The fact that the original distribution
was asymmetric seems not to have much impact on the results. \\

{\em Coverage properties in heavy-tailed regime. }  The general
picture changes dramatically if we investigate the case of
heavy-tailed distributions. Here we consider simulations  with $X \sim $Pareto, $Y \sim$
Paretoinverted. The reason for using the inverted Pareto distribution for $Y$ 
(instead of the ``standard'' one) is that we want
to study a general asymmetric case --- the distribution of the
projections on $L_\rhoperp$ would be perfectly symmetric in case where
both $X$ and $Y$ are generated according to the same
distribution. Results for $X, Y \sim $ Pareto can be found in the
supplement.  
In Figure~\ref{fig-simulations-coverage},
third row, we can see that for the heavy-tailed parameters $a < 2$, both
Fieller's and Hwang's confidence sets fail completely and lead to
empirical coverage probabilities below 0.20 instead of 0.90. For
Hwang, the happens no matter what bootstrap method we use (symmetric or equal-tailed,
bootstrap-t or Hall), see Figure~\ref{fig-simulations-coverage}, rows three
to five and supplement. %
The method
Geometric(equal-tailed, Hall), on the other hand, performs much better
than both Fieller's and Hwang's methods in the heavy-tailed regime
$a<2$. The overall coverage of the geometric method never drops below
0.70, a dramatic improvement over the other two methods. %
It is interesting to observe that the good performance of the
geometric method in the heavy-tailed regime decreases massively if we use
bootstrap-t instead of Hall's bootstrap intervals
(Figure~\ref{fig-simulations-coverage}, fifth row). The reason is
that in the heavy-tailed case, bootstrap-t does not achieve 
good coverage for the one-dimensional projections, and then of course
the coverage of the
final confidence intervals suffers as well. 
Finally, when the Pareto tail parameter moves in the region $a>2$, we are again in the
domain of attraction of the normal law. Here all results
resemble again the ones already reported for the finite variance case. \\

{\em Interpretation of the results in terms of projections. } The
quality of all three methods crucially depends on the quality of the
one-dimensional confidence sets under consideration.  For
distributions in the domain of attraction of the normal law, Fieller's
confidence sets perform very well, even for highly asymmetric
distributions. The reason is that even for small sample sizes, the
distribution of the sample means is already so close to normal that
using bootstrap does not lead to any advantage over using a normal
distribution assumption. 
In the heavy-tailed regime, both Hwang and Fieller fail. This is the
case because both of them do not achieve good coverage probabilities
for the projected one-dimensional random variables $T_{\hat\rho,\hat
  C}(\hat\mu)$ in the first place. 
Here the geometric method has a big advantage over the
other two methods, because instead of considering projections in
arbitrary directions we only have to deal with projections on the
coordinate axes.  The fact that the coverage of the one-dimensional
confidence sets on the projections is an important indicator for the quality of the
confidence set for the ratio can also observed from the fact that the
coverage of 0.70 achieved by
Geometric(Hall) (Figure~\ref{fig-simulations-coverage}, rows three and
four) is in accordance with values reported by
\citet{RomWol99} 
 for the coverage of confidence sets for the mean of 
Pareto distributions. \\

{\em Number of bounded confidence sets. }   In Figure
\ref{fig-simulations-bounded} we compare the number of bounded
confidence sets for the three methods. Often, those numbers do not
differ too much across the different methods. In some cases, 
Geometric(equal-tailed)  performs favorably in that it has more bounded
confidence sets than the other methods (see supplement for more
figures). 
In the asymmetric heavy-tailed case it can be seen that when using symmetric
rather than equal-tailed confidence sets in the geometric method, the number of bounded
confidence sets decreases heavily (compare third and fourth row of
Figure \ref{fig-simulations-bounded}). This is
due to the fact that the one-dimensional confidence sets then become very large in
both directions (whereas the equal-tailed ones are only large in one
direction). Hence, the origin is contained in the resulting rectangle much more
often, which then leads to unbounded confidence sets. This strongly speaks in
favor of using equal-tailed bootstrap confidence sets rather than
symmetric ones in the geometric
method.  
Note that for Hwang's method, using
equal-tailed confidence sets can lead to implausible confidence sets
which are unbounded on one side, but bounded on the other side (as
explained above). In our
experiments, such confidence sets indeed did occur, but not very often
(about 20 times out of 1000 repetitions). \\

{\em Summary. }  The geometric approach to confidence sets for ratios
shows that confidence sets for ratios can be derived from 
one-dimensional confidence sets for the mean of projections of
$(X,Y)$. Of course, the quality of the ratio confidence sets crucially
depends on the quality of those one-dimensional confidence sets.  Based on our experiments, we would like to give the
following advice.  For distributions which are in the domain of
attraction of the normal law, we recommend to use Fieller's confidence
set instead of using any bootstrap method.  Here, Fieller's set works
fine even for small sample size and in asymmetric
distributions.  %
Hwang's set achieves
comparable results in terms of absolute deviation, but as opposed to
Fieller's sets its deviations tend to be to the liberal side, which
should be avoided in our opinion. For asymmetric heavy-tailed
distributions we recommend to use our Geometric(equal-tailed, Hall)
method.  This method can be seen as a natural generalization of the
geometric interpretation of the Fieller method to a bootstrap
scenario.  Even though it does not work perfect, its coverage
outperforms Fieller's and Hwang's methods by a large margin, and the
number of bounded confidence sets is often higher than for Fieller or
Hwang. The performance of the geometric method of course depends on
the performance of the bootstrap method used for the one-dimensional
distributions. If one is able to improve the bootstrap intervals for
the mean of those distributions, one is very likely to further improve
the coverage of
the geometric confidence sets for the ratio. \\

\par\noindent {\normalfont \Large\bfseries Acknowledgments\\}
We would like to thank Volker Guiard for helpful comments on an earlier version of this manuscript.

\bibliography{general_bib,lit-vf,lit-general,lit-ratio-users,ules_publications}

\vskip .65cm

\noindent
Ulrike von Luxburg, 
Max Planck Institute for Biological Cybernetics, 
Spemannstr.~38, 
72076~T{\"u}bingen, 
Germany; 
Phone: +49 7071 601 540; 
Fax:  +49 7071 601 552
\vskip 2pt
\noindent
E-mail: ulrike.luxburg@tuebingen.mpg.de\\

\noindent
Volker H. Franz, Justus-Liebig-Univerit{\"a}t Giessen, FB06, 
Abt. Allgemeine Psychologie, Otto-Behaghel-Str.~10F, 35394~Giessen, Germany
\vskip 2pt
\noindent
E-mail: volker.franz@psychol.uni-giessen.de
\vskip .3cm

\end{document}